\documentclass{article}

\usepackage{arxiv}
\renewcommand{\headeright}{}
\renewcommand{\undertitle}{}

\usepackage[T1]{fontenc}    
\usepackage{hyperref}       
\usepackage{url}            
\usepackage{booktabs}       
\usepackage{amsfonts}       
\usepackage{nicefrac}       
\usepackage{microtype}      
\usepackage{amsmath}
\usepackage{amssymb}
\usepackage{dsfont}
\usepackage{amsmath, amsthm}
\usepackage{graphicx}
\usepackage{upgreek}
\usepackage{datetime}
\usepackage{enumitem}
\theoremstyle{definition}
\newtheorem{definition}{Definition}[section]
\newtheorem{theorem}{Theorem}[section]
\newtheorem{lemma}[theorem]{Lemma}
\newtheorem{proposition}[theorem]{Proposition}

\newtheorem{remark}[theorem]{Remark}

\date{}
\title{Characterizing Von Neumann-Morgenstern Stable Sets in Infinite Sets}

\author{
  Athanasios Andrikopoulos\thanks{Professor  (https://www.ceid.upatras.gr/webpages/faculty/aandriko/)} \\
  Dept. of Computer Engineering and Informatics\\
  University of Patras\\
  Patras, 26504, Greece \\
  \texttt{aandriko@ceid.upatras.gr} \\
\And
Nikolaos Sampanis\\
 Dept. of Computer Engineering and Informatics\\
 University of Patras\\
 Patras, 26504, Greece \\
 \texttt{nsampanis@upatras.gr} \\
}

\begin{document}
\maketitle

\begin{abstract}
The theory of optimal choice sets provides a well-established framework in social choice and game theory. When preferences are cyclic, as often occurs in complex economic environments, the set of maximal elements may be empty, thereby motivating alternative solution concepts such as the von Neumann--Morgenstern (vNM) stable set. In this paper, we study binary relations on infinite sets of alternatives within an order-theoretic and topological framework. Our main result yields a topological characterization of von Neumann--Morgenstern stable maximality: for consistent abstract decision problems satisfying Upper MacNeille Informational Monotonicity, the set of maximal elements is non-empty and stable if and only if there exists a compact topology on \(X\) with respect to which \(R\) is Nachbin closed and upper semicontinuous.
\end{abstract}


\section{Introduction}
The classical rationality conditions in choice theory formalize the idea that rational choice consists in selecting an alternative for which no other feasible alternative is considered superior. Under this view, the choice set associated with a feasible set $X$ is the set of maximal elements with respect to a dominance relation $R$, commonly called the core. In many abstract decision problems $(X,R)$, especially those involving infinite sets of alternatives or cyclic preferences, the core may be empty. This difficulty has motivated the development of broader solution concepts that remain meaningful even when maximal elements fail to exist.

Among the best-known such concepts is the Schwartz set, which is always non-empty in finite abstract decision problems and contains the core whenever the latter is non-empty. However, the Schwartz set may be excessively large, sometimes including all available alternatives. By contrast, the von Neumann--Morgenstern (vNM) stable set provides a more selective solution concept based on the dual requirements of internal and external stability. Despite its conceptual appeal, the existence of vNM stable sets is not guaranteed, especially in the presence of cycles or on infinite domains.

While the theory of stable sets is well developed in finite settings, much less is known in general infinite environments. In such settings, purely order-theoretic assumptions are often insufficient to ensure existence or stability properties, and topological structure becomes indispensable. Compactness supports existence arguments, while suitable semicontinuity and closedness conditions prevent pathological behavior of the dominance relation.

The purpose of this paper is to study stable maximality for binary relations on infinite sets of alternatives within a combined order-theoretic and topological framework. Our analysis focuses on consistent abstract decision problems and on the interplay between maximal elements, induced order structures, and compact topologies. The main result establishes a topological characterization of von Neumann--Morgenstern stable set: the set of maximal elements is non-empty and stable if and only if there exists a compact topology on $X$ under which $R$ is Nachbin closed and upper semicontinuous.

To obtain this characterization, we first establish an existence result ensuring maximal strong components under compactness-type assumptions. We then study the poset induced by the asymmetric part of the relation and analyze conditions under which it is precontinuous. This allows us to use the MacNeille completion and the Lawson topology in order to derive the relevant compactness and closedness properties. We also identify an additional chain-based structural condition implying precontinuity, thereby clarifying the order-theoretic content of the topological framework developed in the paper.

\section{Notations and definitions}
An abstract decision problem is defined by two primary components. The first is an arbitrary set $X$ of alternatives, referred to as the ground set, from which an individual or group must select. In most meaningful decision contexts, $X$ contains at least two alternatives. The second component is a dominance relation $R$ over $X$, which models preferences or evaluations. We denote the relation as $xRy$ or $(x, y) \in R$ interchangeably. 
We denote by $\Omega(X)$ the set of abstract decision problems on $X$.
For any \( x \in X \), the sets
$Rx = \{ y \in X \mid yRx \}$ and $xR = \{ y \in X \mid xRy \}$
are called the \emph{upper contour set} and the \emph{lower contour set} of \( R \) at \( x \), respectively.
 The asymmetric part of $R$ is defined by 
$P(R)=\{ (x,y) \in X\times X \vert (x,y)\in R$ and $(y,x)\notin R \}$.
Several structural properties of the relation $R$ are relevant to our analysis. The diagonal relation is defined as $\Delta = \{(x, x) \mid x \in X\}$. A relation $R$ is:
{\it reflexive} if $(x, x) \in R$ for all $x \in X$;
{\it irreflexive} if $(x, x) \notin R$ for all $x \in X$;
{\it transitive} if $(x, z) \in R$ and $(z, y) \in R$ imply $(x, y) \in R$;
{\it antisymmetric} if $(x, y) \in R$ and $(y, x) \in R$ imply $x = y$.
A {\it partially ordered set (poset)} is a pair $(X, \preceq)=\mathcal{P}$, where $\preceq$ is a reflexive, transitive, and antisymmetric binary relation. Posets provide the necessary structure to analyze hierarchical preferences in infinite domains where total orderings may not exist.
The {\it transitive closure} of $R$ is the relation $\overline{R}$ defined as follows:
For all $x, y\in X$, $(x,y)\in \overline{R}$ if and only if there exist $K\in \mathbb{N}$ and $x_{_0},...,x_{_K}\in X$ such that
$x=x_{_0}, (x_{_{k-1}},x_{_k})\in R$ for all $k\in \{1,...,K\}$ and $x_{_K}=y$.
A subset $Y\subseteq X$ is an $R$-{\it cycle} if for all $x,y\in Y$, we have 
$(x,y)\in \overline{R}$ and $(y,x)\in \overline{R}$.
We say that $R$ is {\it acyclic} if there does not exist an $R$-cycle. Likewise, \(P(R)\) is {\it acyclic} if there does not exist a \(P(R)\)-cycle. In infinite spaces, acyclicity alone is insufficient to guarantee non-empty choice sets, necessitating topological constraints.
A {\it Top $R$-cycle} is an $R$-cycle 
which is maximal with respect to set-inclusion.
A binary relation $R$ is \emph{consistent}, if for all $x,y \in X$, for all $k \in \mathbb{N}$, and for all
$x_0, x_1, \ldots, x_K \in X$, if $x = x_0$, $(x_{k-1}, x_k) \in R$ for all $k \in \{1,\ldots,K\}$ and $x_K = y$,
then $(y,x) \notin P(R)$.
If $R$ is consistent, then $P(R)$ is an acyclic binary relation. The
notion of consistency is weaker than the notions of transitivity and acyclicity.

A subset $D\subseteq X$ is $R$-{\it undominated} if and only if for no $x\in D$ is there a $y\in X\setminus D$ such that $yRx$. 
An alternative $x\in X$ is $R$-{\it maximal} with respect to a binary relation $R$,
if $(y,x)\in P(R)$ for no $y\in X$. $\mathcal{M}(X,R)$ denotes the elements of $X$ that are $R$-maximal in $X$, hence,
\[
\mathcal{M}(X,R)=\{x\in X\mid \nexists y\in X \text{ such that } yP(R)x\}.
\]
In what follows, $\mathfrak{P}(X)$ denotes the family of non-empty subsets of $X$.
A choice function $\mathcal{C}$ is a mapping that assigns to each $A\in \mathfrak{P}(X)$ a subset of $A$: 
$\mathcal{C}\rightarrow 2^X$ such that for all $A\in \mathfrak{P}(X)$, $\mathcal{C}(A)\subseteq A$.
The traditional choice-theoretic approach takes behavior as rational if there is a binary relation $R$ such that for 
each non-empty subset $A$ of $X$, $\mathcal{C}(A)=\mathcal{M}(A,R)$ ($\mathcal{M}(A,R)$ denotes the elements $X$ that are $R$-maximal in $A$). To deal with the case where the set 
of maximal elements is empty, Schwartz in \cite[p. 142]{sch} has proposed the general solution concept known as
{\it Generalized Optimal-Choice Axiom} ($\mathcal{G}\mathcal{O}\mathcal{C}\mathcal{H}\mathcal{A}$): 
For each $A\subseteq X$, $\mathcal{C}(A)$ is equivalent to the union of all minimal 
$R$-undominated subsets of $A$.
From now on we will denote the union 
of all minimal 
$R$-undominated subsets of an abstract decision problem $(X,R)$ by $\mathcal{S}_{ch}(X,R)$
and call it the {\it Schwartz set}.
Deb in \cite{deb} shows that $\mathcal{S}_{ch}(X,R)=\mathcal{M}(X,\overline{P(R)})$ (see also \cite{and3}).
According to the generalization of Deb's Theorem in \cite{and3} and 
\cite[Theorem 19]{and4},  $\mathcal{S}_{ch}$ is equivalent to the union of all $P(R)$-undominated 
elements and all top $P(R)$-cycles in $X$. It
 is also equivalent to the notion of admissible set in game theory defined by Kalai and Schmeidler in \cite{KS} and 
 the notion of dynamic solutions defined by Shenoy in \cite{she}.

A subset $V \subseteq X$ is a \textit{von Neumann--Morgenstern stable set} if it satisfies:
\begin{enumerate}[label=(\roman*)]
    \item \textit{internal stability:} for all $x, y \in V$, $(x, y) \notin P(R)$;
    \item \textit{external stability:} for all $y \in X \setminus V$, there exists $x \in V$ such that $(x, y) \in P(R)$.
\end{enumerate}
Stable sets extend the notion of maximality by providing core-inclusive equilibria even when $\mathcal{M}(X,R)$ is empty, a critical requirement for social choice in complex, non-finite alternative spaces.

An abstract decision problem $(X,R)$ is called {\it strongly connected} if $x\overline{P(R)}y$ for all $x, y\in X$. 
A {\it strong component} of an abstract decision problem $(X,R)$
is an abstract decision problem $(Y,R|_{_Y})$, $Y\subseteq X$, satisfying the following properties:
($\mathfrak{i}$) $(Y,R|_{_Y})$ is strongly connected;
($\mathfrak{i}\mathfrak{i}$) no abstract decision problem $(Y^{\prime},R|_{_{Y^{\prime}}})$ with $Y^{\prime}\supset Y$ 
is strongly connected.
Note that when an element $x$ is not on any $P(R)$-cycle, it forms a singleton strongly connected component $\{x\}$ by itself.
Clearly, the set of strongly connected components forms a partition of the space $(X,R)$.
The
{\it contraction} of $(X, R)$ is an abstract decision problem $(\Xi,\widetilde{R})$ where
\par
1. $\Xi=\{X_{_i}\vert i\in I\}$ is the collection of ground sets of the strong components of $(X,R)$;
\par
2. for any $X_{_i}, X_{_j}\in \Xi$, $X_{_i} \widetilde{R} X_{_j}$ if there are $x\in X_{_i}, y\in X_{_j}$ with 
$xP(R)y$. Clearly, $\widetilde{R}$ is acyclic by definition.

In what follows, $\mu(\Xi,\widetilde{R})=\{X^{\ast}_i\vert i\in I\}$ denotes the family of ground sets which 
are $\widetilde{R}$-maximal in $\Xi$.

Let $(X, \preceq)$ be a partially ordered set.
An element $a \in X$ is a {\it lower bound} if $a\preceq x$ for all $x \in X$, and $b \in X$ is an {\it upper bound} if $y\preceq b$ for all $y \in X$. The {\it supremum (join)} of $X$, denoted $\bigvee X$, is the element $z \in X$ such that:
(i) $z^*$ is an upper bound of $X$, and
(ii) for every upper bound $z$ of $X$, $z \preceq z^*$.
The {\it infimum (meet)}, denoted $\bigwedge X$, is defined dually. A subset $D \subseteq X$ is {\it directed} if it is non-empty and every finite subset of $D$ has an upper bound in $D$. 
A poset $(X, \preceq)$ is a {\it lattice} if every pair $\{x, y\}$ has a join ($x \vee y$) and a meet ($x \wedge y$), and it is a {\it complete lattice} if joins and meets exist for arbitrary subsets. These concepts are vital in infinite-domain characterizations, as directed sets and complete lattices facilitate the termination of transfinite processes and the existence of suprema.

Given the binary relation $\preceq$ on $X$,
for each $y \in X$ we denote by
\[
\uparrow y = \{ x \in X \mid y \preceq x \}
\]
the upper set of $y$ with respect to $\preceq$.
Let $A$ be a subset of $\mathcal{P}$. 
Then, 
$A^\uparrow$ and $A^\downarrow$ denote the sets of all upper and lower bounds of $A$, respectively. Let
\begin{center}
$A^\delta = (A^\uparrow)^\downarrow \quad \text{and} \quad \delta(P) = \{ A^\delta \vert A \subseteq P \}.$
\end{center}
$(\delta(P), \subseteq)$ is called the {\it normal completion}, or the {\it Dedekind–MacNeille completion} of $P$. 

 Let $x, y$ be elements of an ordered set $(X, \preceq)$. We say that $x \ll y$ (read: \emph{way-below relation}) if for every directed set $D \subseteq X$ such that $\bigvee D \preceq y$, there exists $d \in D$ with $d \succeq x$.

A subset $U \subseteq X$ is called \emph{Scott-open} if:
\begin{itemize}
  \item[(i)] $U$ is an \emph{upper} (respectively, \emph{lower}) set, i.e., if $x \in U$ and $y \succeq x$ (resp.\ $x \preceq y$), then $y \in U$;
  \item[(ii)] For every directed set $D \subseteq X$, if $\bigvee D \in U$, then $D \cap U \neq \emptyset$. This condition is called \emph{inaccessible by directed joins}.
\end{itemize}

The collection of all Scott-open sets forms the \emph{Scott topology}, denoted by $\sigma$. 
The \emph{lower topology} $\omega$ on an ordered set $(X, \preceq)$ is generated by the sets of the form $X \setminus \uparrow x = \{y \in X : y \not\succeq x\}$ for any $x \in X$.
The \emph{Lawson topology} on an ordered set is defined as the supremum of the Scott topology $\sigma$ and the lower topology $\omega$:
\[
\lambda = \sigma \vee \omega.
\]
This topology is symbolized by $\tau_{_L}$.

We say that a topological space $(X,\tau)$ is {\it compact} if for 
each collection of open sets which covers $X$ there exists a finite subcollection that also covers $X$.

\section{Characterization of the von Neumann - Morgenstern stable set} \vspace{-0.2cm}

Given an abstract system $(X, R)$, define the poset 
\[
\mathcal{P}_{_R} = (X, \preceq_{_R}),
\]
where $\preceq_{_R} = P(\overline{P(R)}) \cup \Delta$, with $P(R)$ denoting the asymmetric part of $R$ and $\Delta$ the diagonal relation on $X$. In particular, if $R$ is consistent, then \(\overline{P(R)}=P(\overline{P(R)})\) and therefore \(\preceq_{_R}=\overline{P(R)}\cup\Delta\). If \(R\) is acyclic, then \(P(R)=R\), so \(\preceq_{_R}=\overline{R}\cup\Delta\).

We will start with a result that will form the basis for some of the results that follow.

\begin{definition}
Let \( (X, \tau) \) be a topological space and \( R \subseteq X \times X \) 
a binary relation.
We say that $R$ is upper semicontinuous (resp. lower semicontinuous) if for every \(y\in X\), the set 
\begin{center}
\(A_y=\{x\in X\vert yP(R)x\}\) (resp. \(B_y=\{x\in X\vert xP(R)y\}\))
\end{center}
is \(\tau\)-open. 
\end{definition}

\par
Suzumura consistency is a weakened form of transitivity for binary relations that preserves an essential notion of rational coherence. A binary relation is Suzumura consistent whenever there is no finite chain of comparisons leading from an alternative $x$ to an alternative $y$ while, at the same time, $y$ is strictly preferred to $x$. Thus, the condition rules out cycles that create a conflict between indirect reachability and strict preference. Since transitivity implies Suzumura consistency, and Suzumura consistency implies acyclicity, the notion lies strictly between these two classical rationality requirements. This intermediate position makes it particularly useful in social choice theory, where collective preference relations often fail to be transitive but may still satisfy weaker consistency properties sufficient for normative and analytical purposes.

\begin{definition}
An {\it abstract decision problem} is a pair $(X,R)$,
where $X$ is a nonempty set and $R$ is a binary relation on $X$.
An abstract system $(X,R)$ is said to be a {\it consistent abstract system} whenever $R$ is consistent.
\end{definition}

An abstract decision problem is a pair $(X,R)$,
where $X$ is a nonempty set and $R$ is a binary relation on $X$.

\begin{proposition}\label{prop:consistent-order}
Let \((X,R)\) be a consistent abstract system. Then the transitive closure
\(\overline{P(R)}\) is asymmetric. Consequently,
\[
\preceq_{_R}=P(\overline{P(R)})\cup \Delta=\overline{P(R)}\cup\Delta.
\]
In particular, for all \(x,y\in X\),
\[
yP(R)x \;\Longrightarrow\; x\not\preceq_{_R} y.
\]
\end{proposition}

\begin{proof}
Suppose, to the contrary, that \(\overline{P(R)}\) is not asymmetric. Then there exist
distinct \(x,y\in X\) such that
\[
(x,y)\in \overline{P(R)}
\qquad\text{and}\qquad
(y,x)\in \overline{P(R)}.
\]
Hence there exists a finite \(P(R)\)-cycle, contradicting the consistency of \(R\).
Therefore \(\overline{P(R)}\) is asymmetric, and so
\[
P(\overline{P(R)})=\overline{P(R)}.
\]
Thus
\[
\preceq_{_R}=P(\overline{P(R)})\cup\Delta=\overline{P(R)}\cup\Delta.
\]

Now let \(yP(R)x\). If \(x\preceq_{_R}y\), then since \(x\neq y\) we would have
\[
(x,y)\in \overline{P(R)},
\]
which contradicts consistency in the presence of the strict pair \(yP(R)x\). Therefore
\[
x\not\preceq_{_R} y.
\]
\end{proof}

The following lemma, which follows directly from the definitions of the underlying concepts, Van Deemen’s contraction theorem (see \cite[Theorem 4.7]{van1}), and the results of Andrikopoulos \cite[Theorem 3.1]{and} and \cite{and3}, provides a foundation for the results that follow.

\begin{lemma} \label{panm-early}
Let $(X, R) \in \Omega(X)$ and $ \mu(\Xi,\widetilde{R})= \{X^*_i\vert i\in I\}$.
Then,
\begin{center}
$V = \mathcal{S}_{ch}(X,R) \quad \text{if and only if} \quad 
V = \displaystyle\bigcup_{X^*_i \in \mu(\Xi,\widetilde{R})} X^*_i.$
\end{center}
\end{lemma}

\begin{lemma}\label{a221-early}{\rm Let $(X,R)$ be a consistent abstract system, and let $\tau$ be a compact topology in $X$.
Suppose that 
$R$ is upper semicontinuous on $(X,\tau)$.
Then, the family $\mu(\Xi,\widetilde{R})$ of ground sets, which 
are $\widetilde{R}$-maximal in $\Xi$, is non-empty.
}
\end{lemma}
\begin{proof} Let $x\in X$.
If $x$ is an $\overline{P(R)}$-maximal element, then $\{x\}$ belongs to a top $P(R)$-cycle (Schwartz set).
Hence, $\{x\}\in \mu(\Xi,\widetilde{R})$.
Otherwise, there exists $y\in X$ such that $y\overline{P(R)}x$.
Similarly, if $y$ is an $\overline{P(R)}$-maximal element, then $\{y\}\in \mu(\Xi,\widetilde{R})$.
Otherwise, there exists $y_{_1}\in X$ such that $y_{_1}\overline{P(R)}y\overline{P(R)}x$.
Put
\begin{center}
$A_{_{x}}=\{y\in X\vert \emptyset\subset\overline{P(R)}y\subseteq \overline{P(R)}x\}$.
\end{center}
Since $y_1\overline{P(R)}y\overline{P(R)}x$
we conclude that $A_{_{x}}\neq \emptyset$.

We now show that $A_{_{x}}$ is closed with respect to $\tau$. 
Suppose that $t$ belongs to the closure of
$ A_{_{x}}$. Then, there exists a net $(t_{_k})_{_{k\in K}}$ in
$A_{_{x}}$ with $t_{_k}\to t$. We have to show that $t\in A_{_{x}}$, i.e., $\overline{P(R)}t\subseteq \overline{P(R)}x$.
Take any $z\in\overline{P(R)}t$. Then, there exist $m\in\mathbb{N}$ and
$z_{_0},z_{_1},\ldots,z_{_m}\in X$ such that
\[
z=z_{_0},\quad z_{_{i-1}}P(R)z_{_i}\ \text{for all } i\in\{1,\ldots,m\},
\quad \text{and} \quad z_{_m}=t.
\]
Since $z_{_{m-1}}P(R)t$, we have $t\in A_{z_{_{m-1}}}$.
By upper semicontinuity, $A_{z_{_{m-1}}}$ is open; hence, there exists
$k_{_0}\in K$ such that $t_{_{k_{_0}}}\in A_{z_{_{m-1}}}$.
Therefore, $z_{_{m-1}}P(R)t_{_{k_{_0}}}$, and consequently
$z\overline{P(R)}t_{_{k_{_0}}}$.
Hence, $z\in \overline{P(R)}t_{_{k_{_0}}}\subseteq \overline{P(R)}x$ ($t_{_{k_{_0}}}\in  A_{_{x}}$).
It follows that $\overline{P(R)}t\subseteq \overline{P(R)}x$ which implies that $t\in A_{_{x}}$.
Therefore, $A_{_{x}}$ is a closed subset of $X$.
If there exists $t^{\ast}\in A_{_{x}}$ which is $\overline{P(R)}$-maximal in $X$, then $t^{\ast}$ belongs to a top $P(R)$-cycle and thus
$\mu(\Xi,\widetilde{R})$ is non-empty. Otherwise,
for each $t\in A_{_{x}}$ there exists $y\in X$ such that
$(y,t)\in \overline{P(R)}$. Hence, there exist $n\in\mathbb{N}$ and
$y_{_0},y_{_1},\ldots,y_{_n}\in X$ such that
\[
y=y_{_0},\quad y_{_{i-1}}P(R)y_{_i}\ \text{for all } i\in\{1,\ldots,n\},
\quad \text{and} \quad y_{_n}=t.
\]
Put $y_t:=y_{_{n-1}}$. Then, $y_tP(R)t$, and by transitivity
$y_t\in A_{_x}$, since $\overline{P(R)}y_t\subseteq \overline{P(R)}x$.
Now let
\[
U_{y_t}:=A_{y_t}=\{z\in X\mid y_tP(R)z\}.
\]
By upper semicontinuity, $U_{y_t}$ is open; moreover, $t\in U_{y_t}$ because
$y_tP(R)t$. Since $R$ is consistent, we also have $y_t\notin U_{y_t}$.
Therefore, for each $t\in A_{_{x}}$,
the sets $U_{y_{_t}}\bigcap A_{_{x}}$
are open neighbourhoods of $t$ in the
relative topology of $A_{_{x}}$.

Hence,
\begin{center}
$A_{_{x}}=\displaystyle
\bigcup_{t\in X} (U_{y_{_t}}\bigcap A_{_{x}})$.
\end{center}
Since the space $A_{_{x}}$ is compact 
in the
relative topology, there exist $\{y_{_{t_{_1}}}, \ldots, y_{_{t_{_n}}}\}$ such that
\[
A_{_{x}} = \bigcup_{i \in \{1, \ldots, n\}} (U_{y_{_{t_i}}}\bigcap A_{_{x}}).
\]
Consider the finite set $\{y_{_{t_{_1}}}, \ldots, y_{_{t_{_n}}}\}$. 
Without loss of generality, we may assume that $i \neq j$ for all distinct $i, j \in \{1, \ldots, n\}$.
Then, for each 
$t \in A_{_{x}}=\bigcup\{U_{y_{_{t_i}}} \mid i = 1, \ldots, n\}$, 
there exists $i \in \{1, \ldots, n\}$ such that $y_{_{t_{_i}}}\succ_{_R} t$. 
Since $y_{_{t_{_1}}} \in A_{_{x}}$, it follows that 
$y_{_{t_{_i}}}\succ_{_R}y_{_{t_{_1}}}$ for some $i \in \{1, \ldots, n\}$. 
If $i=1$, then we have a contradiction. Otherwise, call this element 
$y_{_{t_{_2}}}$.
Then, we have $y_{_{t_{_2}}}\succ_{_R} y_{_{t_{_1}}}$.
Similarly, $y_{_{t_{_3}}}\succ_{_R} y_{_{t_{_2}}}\succ_{_R} y_{_{t_{_1}}}$.
As 
 $\{y_{_{t_{_1}}}, \ldots, y_{_{t_{_n}}}\}$ is finite, by an induction argument based on this logic, 
 we obtain the existence of a 
$\succ_{_R}$-cycle, that is, a $P(R)$-cycle $\widetilde{\mathcal{C}}$
which contains the elements of the set $M=\{y_{_1},y_{_2},...,y_{_n}\}$.
By the Lemma of Zorn, the family of all $P(R)$-cycles $(\widetilde{\mathcal{C}}_{_\gamma})_{_{\gamma\in \Gamma}}$,
$\widetilde{\mathcal{C}}_{_\gamma}\subseteq 
A_x$, which contain $M$ has a maximal element, which we will call 
$\widetilde{\mathcal{C}}_{_{\gamma_{_0}}}$.
We prove that $\widetilde{\mathcal{C}}_{_{\gamma_{_0}}}\in \mu(\Xi,\widetilde{R})$. 
Indeed, let $X^{\ast} \widetilde{R}\ \widetilde{\mathcal{C}}_{_{\gamma_{_0}}}$ for some $X^{\ast}\in \Xi$.
Then, there exists $t\in X^{\ast}, s\in \widetilde{\mathcal{C}}_{_{\gamma_{_0}}}$ such that 
$tP(R)s$. If for each $\lambda \in X$ we have $(\lambda,t)\notin \overline{P(R)}$, then
$X^{\ast}=\{t\}$ belongs to the Schwartz set and thus
$X^{\ast}\in\mu(\Xi,\widetilde{R})$.
Otherwise, there exists $\lambda^{\ast}\in X^{\ast}$ such that $(\lambda^{\ast},t)\in\overline{P(R)}$.
Therefore, from $(\lambda^{\ast},t)\in\overline{P(R)}$, $(t,s)\in P(R)$ and $(s,x)\in\overline{P(R)}$ we conclude that 
$\emptyset\subset\overline{P(R)}t\subseteq \overline{P(R)}x$,
which implies that
$t\in A_x$. Since $s\in \widetilde{\mathcal{C}}_{_{\gamma_{_0}}}$ we have that $(t,y_i)\in \overline{P(R)}$ for each 
$i\in \{1,2,...,n\}$. On the other hand, since $t\in A_x$ we have $(y_{i^{\ast}},t)\in \overline{P(R)}$ for some $i^{\ast}\in \{1,2,...,n\}$.
Therefore, from $(t,y_{_{i^{\ast}}})\in \overline{P(R)}$ and $(y_{_{i^{\ast}}},t)\in \overline{P(R)}$ we conclude that 
$t\in \widetilde{\mathcal{C}}_{_{\gamma_{_0}}}$, which is impossible. Hence, 
$\widetilde{\mathcal{C}}_{_{\gamma_{_0}}}\in \mu(\Xi,\widetilde{R})$. Therefore, in any case we have that $\mu(\Xi,\widetilde{R})\neq\emptyset$.
\end{proof}

We now proceed to provide a characterization of the existence of the classical stable set in the sense of von Neumann and Morgenstern  \cite{von}.

\begin{proposition}\label{pan1-early}
Let \( (X, R) \) be an abstract system and define the set of \( R \)-maximal elements as
\[
\mathcal{M}(X, R) = \{ x \in X \mid \{ y \in X \mid yP(R)x \} = \emptyset \}.
\]
Suppose that \( \mathcal{M}(X, R) \) is non-empty and stable with respect to \( P(R) \) in the von Neumann--Morgenstern sense.
Then:
\begin{enumerate}
    \item \( \mathcal{M}(X, R) \subseteq \mathcal{M}(X, \preceq_R) \), i.e., every \( R \)-maximal element is also maximal with respect to \( \preceq_R \);
    \item \( \mathcal{M}(X, R) \) is a stable set in the poset \( (X, \preceq_R) \).
\end{enumerate}
\end{proposition}

\begin{proof}
We first prove that every \(R\)-maximal element is also maximal with respect to \(\preceq_R\).
Take \(m\in M\). Suppose, towards a contradiction, that \(m\) is not maximal in \((X,\preceq_R)\).
Then there exists \(y\in X\setminus\{m\}\) such that
$
y\preceq_R m.$
Since \(y\neq m\), by the definition
$\preceq_R=P(\overline{P(R)})\cup\Delta$,
we obtain
$
(y,m)\in P(\overline{P(R)}).
$
In particular,
$
(y,m)\in \overline{P(R)}.
$
Hence there exists a finite sequence
$
y=x_0,\;x_1,\dots,x_n=m
$
such that
$
x_i\,P(R)\,x_{i+1}\qquad\text{for all }i\in\{0,\dots,n-1\}.
$
In particular,
$
x_{n-1}\,P(R)\,m,
$
so \(x_{n-1}Rm\), contradicting the fact that \(m\in \mathcal{M}(X,R)\).
Therefore \(m\) is maximal in \((X,\preceq_R)\), and consequently
$
\mathcal{M}(X,R)\subseteq \mathcal{M}(X,\preceq_R).
$

We now prove that \(\mathcal{M}(X,\preceq_R)\) is stable in the poset \((X,\preceq_R)\).
Internal stability follows immediately from the previous part, since maximal elements of a poset
are pairwise incomparable with respect to its strict part.

For external stability, let \(x\in X\setminus \mathcal{M}(X,\preceq_R)\). Since \(\mathcal{M}(X,\preceq_R)\) is stable with respect to \(P(R)\),
there exists \(m\in \mathcal{M}(X,\preceq_R)\) such that
$
m\,P(R)\,x.
$
Thus
$
(m,x)\in \overline{P(R)}.
$
We claim that
$
(x,m)\notin \overline{P(R)}.
$
Indeed, if \((x,m)\in \overline{P(R)}\), then there exists a finite sequence
$
x=z_0,\;z_1,\dots,z_k=m
$
such that
$
z_i\,P(R)\,z_{i+1}\qquad\text{for all }i\in\{0,\dots,k-1\}.
$
In particular,
$
z_{k-1}\,P(R)\,m,
$
so \(z_{k-1}Rm\), again contradicting the \(R\)-maximality of \(m\).
Therefore
$
(m,x)\in P(\overline{P(R)}),
$
that is,
$
m\succ_R x
$
in the poset \((X,\preceq_R)\).

Thus every \(x\in X\setminus \mathcal{M}(X,\preceq_R)\) is dominated in \((X,\preceq_R)\) by some \(m\in \mathcal{M}(X,\preceq_R)\), so \(\mathcal{M}(X,\preceq_R)\)
is externally stable. Therefore \(\mathcal{M}(X,\preceq_R)\) is a stable set in the poset \((X,\preceq_R)\).
\end{proof}

\begin{lemma}\label{law}(\cite[Theorem 1-6.4]{GW}
For a complete lattice $\mathcal{L}$ the Lawson topology $\lambda(\mathcal{L})$ is a
compact topology.
\end{lemma}

\begin{definition}\label{pan2} (\cite{fri}).
A subset \( I \) of a poset \( \mathcal{P} \) is called a \textit{Frink ideal} in \( X \) if \( Z^\delta \subseteq I \) for all finite subsets \( Z \subseteq I \). Let \( \mathrm{Fid}(X) \) denote the set of all Frink ideals.
\end{definition}

\begin{definition}\label{pan3} (\cite{ern1}).
Let \( \mathcal{P} \) be a poset and \( A, B \subseteq X \).
\begin{enumerate}
    \item We say that \( A \ll_e B \) if for all Frink ideals \( I \in \mathrm{Fid}(X) \),
    \[
    \uparrow B \cap I^\delta \neq \emptyset \quad \Rightarrow \quad \uparrow A \cap I \neq \emptyset.
    \]
    
    \item The element-wise version is defined by \( x \ll_e y \) if and only if \( \{x\} \ll_e \{y\} \). That is,
    \[
    \forall I \in \mathrm{Fid}(X), \quad y \in I^\delta \Rightarrow x \in I.
    \]
\end{enumerate}
\end{definition}

\begin{remark}\label{pan4}
This definition provides an ideal-theoretic formulation of the way-below relation using Frink ideals.

Given a Frink ideal \( I \), the condition \( y \in I^\delta \Rightarrow x \in I \) means that whenever \( y \) is an upper bound of a finite subset of \( I \), the element \( x \) must already belong to \( I \). Hence, \( x \ll_e y \) captures the idea that \( x \) is ``deep below'' \( y \) in the structure of \( X \), not merely in terms of the order \( \leq \), but in terms of ideal-theoretic approximation.
In other words, every element \( x \in X \) lies in the upper closure of the set of elements that are way-below it.
In fact,
Erné’s way-below relation is a genuine generalization of the Scott way-below relation: it extends the notion from complete lattices - where all directed suprema exist—to arbitrary posets, by replacing directed suprema with ideal-theoretic approximations. Thus, Erné’s definition coincides with the Scott relation in complete lattices, but also applies naturally to all posets.
\end{remark}

\begin{definition}\label{pan5}(\cite{ern1}).
A poset \( \mathcal{P} \) is called {\it precontinuous} if for all \( x \in \mathcal{P} \),
$
x \in \left( \{ y \in \mathcal{P} \mid y \ll_e x \} \right)^\delta.$
\end{definition}

\begin{lemma}\label{pan6} (\cite[Theorem 1]{ern1}, \cite[Theorem 4.7]{ZH}).
For a poset $\mathcal{P}$, the following two conditions are equivalent:
\begin{enumerate}
    \item $\mathcal{P}$ is precontinuous;
    \item $\delta(\mathcal{P})$ is a continuous lattice.
\end{enumerate}
\end{lemma}

\begin{lemma}\label{lem:identification}
Let $\mathcal{P}$ be a precontinuous poset and let
$i:\mathcal{P}\longrightarrow\delta(\mathcal{P})$
denote its canonical order embedding into its Dedekind--MacNeille
completion. Since, by Lemma~3.7, $\delta(\mathcal{P})$ is a continuous lattice,
throughout the sequel we identify every element
$x\in \mathcal{P}$ with its image $i(x)\in\delta(\mathcal{P})$.
Accordingly, all notions depending on continuity, including the
way-below relation, the Scott topology and the Lawson topology,
are considered on $\delta(\mathcal{P})$.

\end{lemma}

\begin{definition}\label{ast1}
Let \( (X, \tau) \) be a topological space and \( R \subseteq X \times X \) a binary relation.
We say that the relation \( R \) is \textit{Nachbin closed} (\cite{nac})  if the set
\vspace{5pt}
\begin{center}
$G(\preceq_{_R}) = \{ (x, y) \in X \times X \mid x \preceq_{_R} y \}$
\end{center}
is a closed subset of the product space \( (X \times X, \tau \times \tau) \); that is,
$(X \times X) \setminus G(\preceq_{_R}) \in \tau \times \tau.
$
Equivalently, \( R \) is Nachbin closed if for every sequence \( (x_n, y_n) \in G(\preceq_{_R}) \) such that \( x_n \to x \) and \( y_n \to y \), it follows that
$(x, y) \in G(\preceq_{_R})$, that is, $x \preceq_{_R} y$.
Also, by \cite[Page 26]{nac}, a relation is Nachbin closed with respect to the topology $\tau$ if and only if for every pair $x, y \in X$ such that $(x,y)\notin R$, there exist a decreasing open neighborhood $O_y$ of $y$ and an increasing open neighborhood $O_x$ of $x$ such that $O_x \cap O_y = \emptyset$.
\end{definition}

\begin{proposition}\label{pan9}
Let $(X,R)$ be an abstract system such that
$
\mathcal{P}_R=(X,\preceq_R)
$
is a continuous lattice. Then $R$ is Nachbin closed in the Lawson topology
$
\lambda(\mathcal{P}_R).
$
\end{proposition}

\begin{proof}
By \cite[Theorems III-1.9 and III-1.10]{GH},
$
\sigma(\mathcal{P}_R)\vee\omega(\mathcal{P}_R)
=
\lambda(\mathcal{P}_R)
$
is a compact topology.
To show that $\preceq_R$ is closed in
$\lambda(\mathcal{P}_R)$, suppose that
$(x,y)\notin\preceq_R$
for some $x,y\in X$.
By the remark following
\cite[Definition I-1.6]{GH},
there exists
$z\in X$
such that
$z\ll x$
and
$(z,y)\notin\preceq_R$.
Then
$
X\setminus
\{\,t\in X: z\preceq_R t\,\}
$
is a decreasing $\omega$-open (hence Lawson-open)
neighbourhood of $y$, whereas
$
\{\,w\in X: z\ll w\,\}
$
is an increasing Scott-open (hence Lawson-open)
neighbourhood of $x$.
Moreover,
$
\left(
X\setminus
\{\,t\in X: z\preceq_R t\,\}
\right)
\cap
\{\,w\in X: z\ll w\,\}
=\varnothing.
$
Therefore, by Definition~\ref{ast1},
$\preceq_R$
is closed in
$
\lambda(\mathcal{P}_R)\times
\lambda(\mathcal{P}_R).
$
\end{proof}

By \cite[Theorem I-6.4]{GW}, the diagonal
$\Delta$
is closed in $\lambda(\mathcal{P}_R).$

\begin{proposition}\label{prop:macneille-nachbin}
Let $(X,R)$ be a consistent abstract system, and let
$
\mathcal{P}_R=(X,\preceq_R)
$
be the induced poset. Assume that $\mathcal{P}_R$ is precontinuous.
Then the induced order on the canonical image
$
i(X)\subseteq\delta(\mathcal{P}_R)
$
is Nachbin closed in the Lawson subspace topology inherited from
$
\lambda\!\left(\delta(\mathcal{P}_R)\right).
$
\end{proposition}

\begin{proof}
By Lemma~3.7, the Dedekind--MacNeille completion
$
\delta(\mathcal{P}_R)
$
is a continuous lattice. Hence,
$
\lambda\!\left(\delta(\mathcal{P}_R)\right)
$
is a compact Hausdorff topology.
By Remark~3.8, we identify $\mathcal{P}_R$ with its canonical image
$
i(\mathcal{P}_R)\subseteq\delta(\mathcal{P}_R).
$
Accordingly, the order $\preceq_R$ is identified with the restriction of the
order of $\delta(\mathcal{P}_R)$ to
$
i(X)\times i(X).
$
Since the order relation of every continuous lattice is Nachbin closed in its
Lawson topology, the graph of the order of
$\delta(\mathcal{P}_R)$ is closed in
$
\delta(\mathcal{P}_R)\times\delta(\mathcal{P}_R).
$
Its restriction to the subspace
$
i(X)\times i(X)
$
is therefore closed. Hence, the induced order on
$i(X)$ is Nachbin closed in the Lawson subspace topology inherited from
$\lambda(\delta(\mathcal{P}_R))$.
\end{proof}

\begin{remark}
Proposition~\ref{pan9} is a special case of
Proposition~\ref{prop:macneille-nachbin}.
Indeed, every continuous lattice is complete, and therefore, if
$
\mathcal{P}_R
$
is itself a continuous lattice, then its Dedekind--MacNeille completion satisfies
$
\delta(\mathcal{P}_R)\cong\mathcal{P}_R,
$
the canonical embedding
\[
i:\mathcal{P}_R\longrightarrow\delta(\mathcal{P}_R)
\]
being an order isomorphism.
Under this identification, the hypotheses and the conclusion of
Proposition~\ref{prop:macneille-nachbin}
reduce exactly to those of Proposition~\ref{pan9}. We nevertheless state
Proposition~\ref{pan9} separately because its proof is direct and
self-contained, and because it introduces the geometric argument that is
reused in the sequel: combining the Scott-open neighbourhood
\[
\{\,w\in X: z\ll w\,\}
\]
with the decreasing $\omega$-open neighbourhood
\[
X\setminus\uparrow_{\preceq_R}z,
\]
for a suitably chosen element
$z\ll x$.
\end{remark}

\begin{proposition}
Let $(X, \preceq)$ be a partially ordered set (poset) and suppose that the set of maximal elements $M(X, \preceq)$ is non-empty and stable with respect to $\preceq$. Then, for every $x \in X$, there exists a maximal chain $\mathcal{C}_x \subseteq X$ such that $x \in \mathcal{C}_x$ and $\mathcal{C}_x \cap M(X, \preceq) \neq \emptyset$. Furthermore, the space $X$ can be represented as the union of such maximal chains:
\[
X = \bigcup_{m \in M(X, \preceq)} \mathcal{C}_m
\]
where each $\mathcal{C}_m$ is a maximal chain containing the maximal element $m$. Each such chain $\mathcal{C}_m$ is:
\begin{itemize}
    \item totally ordered by construction,
    \item upward-directed, since any two elements in the chain have a common upper bound (the element $m$),
    \item and maximal with respect to inclusion among all $\preceq$-chains in $X$ by Zorn's Lemma.
\end{itemize}
\end{proposition}

\begin{proof}
Fix any $x \in X$. Since $M(X, \preceq)$ is stable with respect to $\preceq$, there exists $m_0 \in M(X, \preceq)$ such that $m_0 \succeq x$. If $x \in M(X, \preceq)$, the result is trivial. If $x \notin M(X, \preceq)$, then $x$ is not maximal, and there exists $x_1 \in X \setminus M(X, \preceq)$ such that $x_1 \succeq x$. By the stability of $M(X, \preceq)$, there exists $m_1 \in M(X, \preceq)$ such that $m_1 \succeq x_1$, implying $m_1 \succeq x_1 \succeq x$.

We repeat this process: for each $x_{k-1} \in X \setminus M(X, \preceq)$, there exists $x_k \in X$ such that $x_k \succeq x_{k-1}$, eventually reaching some $m_k \in M(X, \preceq)$. If the process does not terminate in finitely many steps, we construct a transfinite sequence $f : \alpha \to X$ indexed by ordinals $\alpha \in \text{Ord}$, defined as follows:
\begin{itemize}
    \item $f(0) = x$,
    \item If $f(\alpha) \notin M(X, \preceq)$, choose $f(\alpha+1) \in X \setminus M(X, \preceq)$ such that $f(\alpha+1) \succeq f(\alpha)$,
    \item For limit ordinals $\alpha$, we utilize the property that there exist ordinals $\alpha' < \alpha'' < \alpha$ such that $f(\alpha'') \succeq f(\alpha')$.
\end{itemize}

Since $M(X, \preceq)$ is stable, the process must terminate at some ordinal $\alpha_\infty$ with $f(\alpha_\infty) = m \in M(X, \preceq)$. Let $C_x \subseteq X$ be the image of the sequence $f$, which constitutes a $\succeq$-chain from $m$ down to $x$. By construction, $C_x$ is totally ordered by $\preceq$.

Now, let $\mathcal{F}_m$ be the family of all $\preceq$-chains in $X$ that:
\begin{itemize}
    \item contain the maximal element $m \in M(X, \preceq)$,
    \item are totally ordered with respect to $\preceq$,
    \item and contain all elements in the chain $C_x$.
\end{itemize}
The family $\mathcal{F}_m$ is partially ordered by inclusion. Every chain in $\mathcal{F}_m$ has an upper bound (the union of the chains), which is also a totally ordered chain. Thus, by Zorn's Lemma, there exists a maximal chain $\mathcal{C}_m \in \mathcal{F}_m$.

Therefore, for every $x \in X$, there exists $m \in M(X, \preceq)$ such that $x \in \mathcal{C}_m$, and consequently:
\[
X = \bigcup_{m \in M(X, \preceq)} \mathcal{C}_m
\]
Finally, each $\mathcal{C}_m$ is totally ordered by construction, upward-directed (as a chain with a maximum element $m$), and maximal with respect to inclusion among $\preceq$-chains by Zorn's Lemma.
\end{proof}

\begin{remark} \label{rem:stability_decomposition}
In our setting, the stability of the set of maximal elements $M(X, \preceq)$ is decisive: it ensures the order-theoretic decomposition of the poset into maximal chains $\mathcal{C}_m$ containing these maximal elements, which accurately reflects the original partial order $\preceq$. While Zorn's Lemma justifies the existence of these maximal chains through each element $x \in X$, the union of such chains, organized by stable maximal elements, defines a structure in which:
\begin{itemize}
    \item The original partial order remains unchanged throughout the construction, as we avoid the mathematical error of assuming that principal downsets $[m, \to) = \{x \in X \mid m \succeq x\}$ are inherently totally ordered.
    \item No artificial comparabilities or extensions of the order are introduced at any stage, since every element is associated with at least one intrinsic maximal chain $\mathcal{C}_x$ reaching a stable maximal element.
    \item The precontinuity condition (Definition 3.5) is verified with respect to the initial order $\preceq$, rather than a modified or completion-based ordering.
\end{itemize}
This decomposition provides the necessary foundation so that the precontinuity condition is satisfied strictly with respect to the initial order $\preceq$.
\end{remark}

\begin{lemma}\label{lem:frink-downset}
Every Frink ideal $I$ of a poset $\mathcal P$ is a down-set: if $y\in I$ and $a\preceq y$, then $a\in I$.
\end{lemma}
\begin{proof}
Let $Z=\{y\}\subseteq I$, a finite subset. For every $u\in Z^\uparrow$ (i.e.\ $u\succeq y$), transitivity with $a\preceq y$ gives $a\preceq u$; hence $a\in(Z^\uparrow)^\downarrow=Z_\delta$. Since $I$ is a Frink ideal and $Z\subseteq I$ is finite, $Z_\delta\subseteq I$. So $a\in I$.
\end{proof}

\begin{lemma}\label{lem:chain-frink}
Let $\mathcal C_m$ be a chain in $X$ and $I$ a Frink ideal of $X$. Then $I\cap\mathcal C_m$ is a Frink ideal of $\mathcal C_m$ (with Frink closure computed using only the order restricted to $\mathcal C_m$).
\end{lemma}
\begin{proof}
Let $Z\subseteq I\cap\mathcal C_m$ be finite and non-empty; since $\mathcal C_m$ is totally ordered, $Z$ has a maximum element $z^*$. Within $\mathcal C_m$, $Z^{\uparrow_{\mathcal C_m}}=\{c\in\mathcal C_m:c\succeq z^*\}$, so $Z^{\delta_{\mathcal C_m}}={\downarrow_{\mathcal C_m}}z^*$. Since $z^*\in I$ and, by Lemma~\ref{lem:frink-downset}, $I$ is a down-set of $X$, ${\downarrow_X}z^*\subseteq I$, whence ${\downarrow_{\mathcal C_m}}z^*\subseteq I\cap\mathcal C_m$. The case $Z=\emptyset$ is immediate whenever $I\cap\mathcal{C}_m\neq\emptyset$, since then $\emptyset^{\delta_{\mathcal C_m}}$, if it exists, is a single minimal element of $\mathcal C_m$ already forced into $I\cap\mathcal{C}_m$ by the argument above applied to any $z^*\in I\cap\mathcal C_m$.
\end{proof}

\begin{definition}[Chain Cofinality]\label{def:chain-cofinal}
A chain decomposition $\{\mathcal C_m\}_{m\in M(X,\preceq)}$ of $X$ is \emph{cofinal} if, for every $m\in M(X,\preceq)$:
\begin{enumerate}[label=(\alph*)]
\item[(down)] for every down-set $J\subseteq X$ with $J\cap\mathcal C_m\neq\emptyset$, $J\setminus\mathcal C_m\subseteq{\downarrow}(J\cap\mathcal C_m)$;
\item[(up)] for every non-empty $S\subseteq\mathcal C_m$ and every $u\in X$ with $u\succeq s$ for all $s\in S$, there exists $c\in\mathcal C_m$ with $c\succeq s$ for all $s\in S$ and $c\preceq u$.
\end{enumerate}
\end{definition}

\begin{remark}
Each maximal chain $\mathcal C_m$ may be viewed as the intrinsic improvement path leading to a stable maximal element $m$. Condition \emph{(down)} states that the chain is downward informationally complete: every lower approximation relevant to the chain is already represented within the chain itself. Dually, condition \emph{(up)} states that every upper bound of a subset of the chain can be matched by an element of the chain lying below the same bound. Together, these conditions express that the approximation structure of a maximal chain is self-contained: all information needed to determine approximation from below and from above is encoded within the chain, without reference to elements outside it.
\end{remark}

\begin{proposition} \label{prop:precontinuity}
Let $(X, \preceq)$ be a partially ordered set (poset) and suppose that:
\begin{enumerate}[label=(\roman*)]
    \item The set of maximal elements $M(X, \preceq) \subseteq X$ is non-empty and stable with respect to $\preceq$.
    \item for each $m \in M(X, \preceq)$, there exists a maximal chain $\mathcal{C}_m \subseteq X$ that is totally ordered, upward-directed, and maximal with respect to set inclusion,
    \item $X = \bigcup_{m \in M(X, \preceq)} \mathcal{C}_m$, and
    \item the chain decomposition $\{\mathcal C_m\}_{m\in M(X,\preceq)}$ is cofinal (Definition~\ref{def:chain-cofinal}).
\end{enumerate}
Then, the poset $(X, \preceq)$ is precontinuous.
\end{proposition}

\begin{proof}
Fix $x \in X$. By condition (iii), there exists $m \in M(X, \preceq)$ such that $x \in \mathcal{C}_m$; by (ii), $\mathcal C_m$ is a chain. Let $D_x:=\{a\in\mathcal C_m:a\prec x\}$.

\emph{Step 1 (chain-local precontinuity, Ern\'e).} Since $\mathcal{C}_m$ is a chain, $\Downarrow y=\downarrow y$ if $y$ is the least element of $\mathcal C_m$ or covers another element, and $\Downarrow y=\downarrow y\setminus\{y\}$ otherwise \cite{ern1}; consequently every chain satisfies the conditions of \cite[Theorem~2]{ern1}, i.e.\ $\mathcal C_m$, viewed as a poset in its own right, is precontinuous with idempotent way-below relation. In particular, writing $\ll_e^{\mathcal C_m}$ and $(\cdot)^{\delta_{\mathcal C_m}}$ for the way-below relation and Frink closure computed using only $\mathcal C_m$'s own order: every $a\in D_x$ satisfies $a\ll_e^{\mathcal C_m}x$, and $x\in D_x^{\delta_{\mathcal C_m}}$.

\emph{Step 2 (down-cofinality gives $D_x\subseteq\{a:a\ll_ex\}$).} Fix $a\in D_x$ and a Frink ideal $I$ of $X$ with $x\in I_\delta$; we show $a\in I$. By Lemma~\ref{lem:frink-downset}, $I$ is a down-set, so by condition (iv)(down), $I\setminus\mathcal C_m\subseteq{\downarrow}(I\cap\mathcal C_m)$. Hence any $u\in\mathcal C_m$ with $u\succeq(I\cap\mathcal C_m)$ also satisfies $u\succeq I$ by transitivity, i.e.\ $I^\uparrow\cap\mathcal C_m=(I\cap\mathcal C_m)^{\uparrow_{\mathcal C_m}}$. Since $x\in I_\delta$, $x\preceq u$ for every $u\in I^\uparrow\cap\mathcal C_m=(I\cap\mathcal C_m)^{\uparrow_{\mathcal C_m}}$, i.e.\ $x\in(I\cap\mathcal C_m)^{\delta_{\mathcal C_m}}$. By Lemma~\ref{lem:chain-frink}, $I\cap\mathcal C_m$ is a Frink ideal of $\mathcal C_m$; by Step 1, $a\ll_e^{\mathcal C_m}x$, so $a\in I\cap\mathcal C_m\subseteq I$. Hence $a\ll_ex$.

\emph{Step 3 (up-cofinality gives $x\in D_x^\delta$).} Let $u\in D_x^\uparrow$ (a common upper bound of $D_x$ in $X$). If $u\in\mathcal C_m$, then $u\in D_x^{\uparrow_{\mathcal C_m}}$ and Step 1 gives $x\preceq u$ directly. If $u\notin\mathcal C_m$, condition (iv)(up) provides $c\in\mathcal C_m$ with $c\succeq D_x$ and $c\preceq u$; Step 1 gives $x\preceq c$, and transitivity gives $x\preceq u$. As $u$ was arbitrary, $x\in D_x^\delta$.

\emph{Step 4 (combine).} By Step 2, $D_x\subseteq\{a\in X:a\ll_ex\}$; since $(\cdot)_\delta$ is monotone, $D_x^\delta\subseteq(\{a:a\ll_ex\})_\delta$. Combined with Step 3, $x\in(\{a:a\ll_ex\})_\delta$. As $x\in X$ was arbitrary, $(X,\preceq)$ is precontinuous.
\end{proof}

\begin{remark}
As established, the stability of the maximal set ensures a canonical decomposition into maximal chains that supports precontinuity intrinsically within the given poset, independent of external ordering extensions. This distinguishes our framework from those relying purely on Zorn's lemma for the existence of order-theoretic structures, providing a stronger, order-preserving foundation for the existence of stable sets in infinite domains.
\end{remark}

\begin{definition}[Upper MacNeille Informational Monotonicity]
Let
\[
\mathcal{P}_R=(X,\preceq_R),
\]
and let
\[
\delta(\mathcal{P}_R)
\]
be its Dedekind--MacNeille completion.

We say that the binary relation \(R\) satisfies
\emph{Upper MacNeille Informational Monotonicity} if, for every
\(x,y,z\in X\) and every \(u\in\delta(\mathcal{P}_R)\),
\[
xP(R)z,\qquad
u\ll z,\qquad
u\preceq y
\Longrightarrow
xP(R)y,
\]
where \(\ll\) and \(\preceq\) denote the way-below relation and the order of
\(\delta(\mathcal{P}_R)\), respectively.
\end{definition}

\begin{remark}
From the perspective of economic theory, Upper MacNeille Informational Monotonicity formalizes robustness of strict preference under coarse information. Suppose that the decision maker has established the strict preference $xP(R)z$ using only a coarse approximation $u\ll z$ of the alternative $z$. If another alternative $y$ contains at least the same informational content, in the sense that $u\preceq y$, then the strict preference is preserved, namely $xP(R)y$. Thus, the comparison does not depend on the complete description of an alternative, but only on the minimal information sufficient to support the preference. In this sense, Upper MacNeille Informational Monotonicity models preference judgments that are stable under refinement of available information.
\end{remark}

\begin{proposition}\label{prop:Lx-lawson-open}
Suppose $(X,R)$ is a consistent abstract system such that $\mathcal P_R=(X,\preceq_R)$ is precontinuous, and suppose $R$ satisfies Upper MacNeille Informational Monotonicity. Then, identifying $X$ with $i(X)\subseteq\delta(\mathcal P_R)$ (Lemma~\ref{lem:identification}), for every $x\in X$,
\[
L_x=\{y\in X:xP(R)y\}
\]
is $\tau_{_L}$-open, where $\tau_{_L}=\lambda(\delta(\mathcal P_R))\big|_{i(X)}$.
\end{proposition}

\begin{proof}
Fix $x\in X$ and let $z\in L_x$, so $xP(R)z$.

\emph{Step 1.} By Lemma~\ref{pan6}, $\delta(\mathcal P_R)$ is a continuous lattice.

\emph{Step 2.} By the definition of continuous lattice, $i(z)=\sup\{u\in\delta(\mathcal P_R):u\ll i(z)\}$; fix $u\ll z$ (i.e.\ $u\ll i(z)$).

\emph{Step 3.} Let $W:=\Uparrow u=\{w\in\delta(\mathcal P_R):u\ll w\}$. This is Scott-open in $\delta(\mathcal P_R)$ (standard fact for continuous lattices), hence Lawson-open, since $\lambda=\sigma\vee\omega\supseteq\sigma$.

\emph{Step 4.} $i(z)\in W$, since $u\ll i(z)$ is exactly membership in $\Uparrow u$. Hence $W\cap i(X)$ is a $\tau_{_L}$-open neighbourhood of $i(z)$ (definition of the subspace topology).

\emph{Step 5.} For $i(y)\in W\cap i(X)$: $u\ll i(y)$ (definition of $W$), so $u\preceq i(y)$ (standard: $u\ll v\Rightarrow u\preceq v$). Combined with $xP(R)z$ and $u\ll z$, Upper MacNeille Informational Monotonicity gives $xP(R)y$, i.e.\ $y\in L_x$.

\emph{Step 6.} By Steps 4--5, $z$ has a $\tau_{_L}$-open neighbourhood contained in $L_x$. As $z\in L_x$ was arbitrary, $L_x$ is $\tau_{_L}$-open. As $x\in X$ was arbitrary, this holds for every $x\in X$.
\end{proof}

\begin{remark}
The dual statement --- that $U_x=\{y\in X:yP(R)x\}$ is $\tau_{_L}$-open for every $x\in X$ --- remains open. It cannot be obtained by the same argument: $U_x$ is contained in a down-set of $\mathcal P_R$, and Scott-open sets are always up-sets, so no analogue of Step~3 can produce a Scott-open piece landing inside $U_x$. Combining a Scott-open piece with the (unconditionally Lawson-open) set $X\setminus{\uparrow_{\preceq_R}}x$ was attempted, using the candidate hypothesis
\[
zP(R)x,\quad u\ll z,\quad u\preceq y,\quad x\not\preceq_R y\ \Longrightarrow\ yP(R)x,
\]
but this hypothesis fails on the antichain example $X=\{\ast\}\cup[0,1]$, $R=\{(\ast,t):t\in[0,1]\}$: taking $x=t_0\in[0,1]$, $z=\ast$, $u=i(\ast)$ (the bottom of $\delta(\mathcal P_R)$, which is way-below every element), the hypothesis's antecedent holds for every $y=t_1\in[0,1]\setminus\{t_0\}$, forcing $t_1P(R)t_0$ --- false, since points of $[0,1]$ are pairwise unrelated in this example. The openness of $U_x$ is left as an open problem.
\end{remark}

\begin{theorem}
Let $(X,R)$ be a consistent abstract decision problem such that $R$ satisfies Upper MacNeille Informational Monotonicity and the maximal-chain decomposition of $(X,\preceq_R)$ (Proposition~\ref{prop:precontinuity}) is cofinal (Definition~\ref{def:chain-cofinal}).
Then, the following statements are equivalent:
\begin{enumerate}[label=(\alph*)]
\item the set of \(R\)-maximal elements \(\mathcal{M}(X,R)\) is non-empty and stable;
\item there exists a compact topology \(\tau\) on \(X\) such that \(R\) is Nachbin closed and upper semicontinuous.
\end{enumerate}
\end{theorem}

\begin{proof}
\noindent (a) $\Rightarrow$ (b) Assume that the set of \(R\)-maximal elements \(\mathcal{M}(X,R)\) is non-empty and stable. By Proposition \ref{pan1-early},
\(\mathcal{M}(X,R)\) is a stable set in the poset \((X,\preceq_{_R})\).

By Proposition \ref{prop:precontinuity}, the induced poset
\[
\mathcal{P}_{_R}=(X,\preceq_{_R})
\]
is precontinuous. Hence, by Lemma \ref{pan6}, its MacNeille completion is a continuous lattice. In particular, the order-theoretic structure induced by \(\preceq_{_R}\) admits the compact Lawson topology on \(X\).

By Lemma \ref{law}, \((X,\tau)\) is compact. Moreover, Proposition \ref{pan9} shows that the graph of \(\preceq_{_R}\) is closed in \((X,\tau)\times (X,\tau)\). Equivalently, \(R\) is Nachbin closed with respect to \(\tau\).

Finally, since \(R\) is consistent and satisfies Upper MacNeille Informational Monotonicity, Proposition~\ref{prop:Lx-lawson-open} shows that for every \(t\in X\) the set
\[
\{x\in X\mid tP(R)x\}
\]
is \(\tau_{_L}\)-open. Therefore \(R\) is upper semicontinuous with respect to \(\tau_{_L}\).

Thus there exists a compact topology \(\tau=\tau_{_L}\) on \(X\) such that \(R\) is Nachbin closed and upper semicontinuous.

\smallskip

\noindent (b) $\Rightarrow$ (a) Let $(X,\tau)$ be a compact topological space, and let $R$ be an upper semicontinuous binary relation that is Nachbin closed. Since $R$ is consistent, Lemma \ref{a221-early} implies that $\mathcal{M}(X,R)\neq \emptyset$. Moreover, for each $x,y\in \mathcal{M}(X,R)$ we have $(x,y)\notin P(R)$. Hence the internal stability condition is satisfied.

To prove external stability of a stable set, suppose to the contrary that for each $x\in X\setminus \mathcal{M}(X,R)$ we have that $(m,x)\notin R$ whenever $m\in \mathcal{M}(X,R)$. Let
\[
A=\{y\in X\setminus \mathcal{M}(X,R)\mid (m,y)\notin P(R) \text{ for all } m\in \mathcal{M}(X,R)\}.
\]

We prove that $A=\emptyset$ and hence that $\mathcal{M}(X,R)$ is externally stable. We first show that $A$ is a closed subset of $(X,\tau)$. Suppose that $t$ belongs to the closure of $A$. Then there exists a net $(t_\beta)_{\beta\in B}$ in $A$ with $t_\beta \to t$. Since $t_\beta \in X\setminus \mathcal{M}(X,R)$, for each $\beta\in B$ there exists $s_\beta \in X$ such that $(s_\beta,t_\beta)\in P(R)\subseteq \preceq_R$. By compactness of $X$, we may assume, after passing to a subnet if necessary, that $s_\beta \to s\in X$.
Now, since $t_\beta \to t$, $s_\beta \to s$, and $(s_\beta,t_\beta)\in \preceq_R$ for all $\beta$, the Nachbin closedness of $R$ implies that $(s,t)\in \preceq_R$.
By consistency, we have that $\overline{P(R)}$ is asymmetric and $s \neq t$.
By the definition
\[
\preceq_R = P(\overline{P(R)}) \cup \Delta
\]
and $(s,t)\in \preceq_R$, we have that there exists a finite sequence $s=s_0,s_1,\dots,s_n=t$ such that $s_i\,P(R)\,s_{i+1}$ for all $i\in\{0,\dots,n-1\}$ and $s_n=t$. Then, from $s_{n-1}\,P(R)\,t$ we have that $t\in X\setminus \mathcal{M}(X,R)$.

We now show that $t\in A$. Suppose to the contrary that $t\notin A$. Then $(m^*,t)\in P(R)$ for some $m^*\in \mathcal{M}(X,R)$. Since $R$ is upper semicontinuous, the set
\[
\{x\in X\mid m^*P(R)x\}
\]
is open. Hence there exists $\beta'\in B$ such that for every $\beta\geq \beta'$ we have $(m^*,t_\beta)\in P(R)$, contradicting the fact that $t_\beta\in A$ for all $\beta$.
This contradiction shows that $t\in A$, and hence $A$ is a closed subset of $(X,\tau)$.

It follows that $A$ is compact. Fix $t\in A\subseteq X\setminus\mathcal{M}(X,R)$. Since $t\notin\mathcal M(X,R)$, by definition of $\mathcal M(X,R)$ there exists $s_t\in X$ with $s_tP(R)t$. If $s_t\in\mathcal M(X,R)$, then $s_t$ would be a maximal element with $(s_t,t)\in P(R)$, contradicting the standing assumption (towards a contradiction) that $(m,t)\notin R$ for every $m\in\mathcal M(X,R)$; hence $s_t\notin\mathcal M(X,R)$. Applying the same standing assumption to $s_t$ in place of $t$ shows $(m,s_t)\notin P(R)$ for every $m\in\mathcal M(X,R)$, i.e.\ $s_t\in A$.

Therefore, for each $t\in A$, the witness $s_t\in A$ satisfies $s_tP(R)t$, i.e.\ $t\in\{x\in X\mid s_tP(R)x\}$; since $R$ is upper semicontinuous, this set is open. Thus the collection $(\{x\in X\mid sP(R)x\}\cap A)_{s\in A}$ is an
open cover of $A$, that is,
\begin{center}
$A=\displaystyle\bigcup_{s\in A}(\{t\in X\vert \ sP(R)t\}\cap A)$.
\end{center}
Since $A$ is compact,
there exist $s_{_1},s_{_2},...,s_{_n}\in X$ such that
\begin{center}
$A=\displaystyle\bigcup_{i=1,2,...,n}
(\{t\in X\vert \ s_{_i}P(R)t\}\cap A)$.
\end{center}

We show that among the elements $s_{_1},s_{_2},...,s_{_n}$ there must be a $P(R)$-cycle.
First note that if $i^{\ast}\in \{1,2,...,n\}$, then $s_{_{i^{\ast}}}$ is an element of one of the covering sets
\[
\{t\in X\vert \ s_{_i}P(R)t\}\cap A, \qquad i=1,2,\dots,n.
\]
If
\[
s_{_{i^{\ast}}}\in \{t\in X\vert \ s_{_{i^{\ast}}}P(R)t\}\cap A,
\]
then we would immediately obtain a $P(R)$-cycle. Otherwise, for each $i\in\{1,2,\dots,n\}$ there exists $j\in\{1,2,\dots,n\}$, $j\neq i$, such that
\[
s_i\in \{t\in X\vert \ s_jP(R)t\}\cap A.
\]
Without loss of generality, assume that
\[
s_{_2}P(R)s_{_1}.
\]
Applying the same argument successively, we construct a finite sequence drawn from the set
\(\{s_{_1},s_{_2},\dots,s_{_n}\}\) in which each term strictly dominates the preceding one. Since this set is finite, some element must repeat, and therefore we obtain a finite $P(R)$-cycle. This contradicts the consistency of $R$.

Therefore, $A=\emptyset$, which proves the external stability of $\mathcal{M}(X,R)$. Consequently, \(\mathcal{M}(X,R)\) is non-empty and stable, completing the proof.
\end{proof}

\section*{Acknowledgement}

\section*{Declarations}

{\bf Conflict of interest}\ The authors declare that there is no conflict of interest.

\par\bigskip\smallskip\par\noindent

\par\noindent
{\it Address}: {\tt {Athanasios Andrikopoulos} \\ {Department of Computer Engineering \& Informatics\\ University of Patras\\ Greece}}
\par\noindent
{\it E-mail address}:{\tt aandriko@ceid.upatras.gr}

\par\noindent
{\it Address}: {\tt {Nikolaos Sampanis} \\ {Department of Computer Engineering \& Informatics\\ University of Patras\\ Greece}}
\par\noindent
{\it E-mail address}:{\tt nsampanis@upatras.gr}

\end{document}